\documentclass[11pt,a4paper]{article}
\usepackage{amsfonts}
\usepackage{amssymb}
\usepackage{amsmath}
\usepackage{cite,enumerate,amsthm,enumerate}

\usepackage{lscape}
\usepackage{diagbox}
\usepackage{multirow}
\usepackage{bigstrut}
\usepackage{float}

\newtheorem{theorem}{Theorem}[section]
\newtheorem{lemma}[theorem]{Lemma}

\newtheorem{corollary}[theorem]{Corollary}

\newtheorem{example}[theorem]{Example}
\newtheorem{remark}[theorem]{Remark}

\usepackage{bigstrut}
\usepackage{diagbox}
\usepackage{multirow}
\usepackage{mathdots}

\begin{document}

\setcounter{page}{1}

\markboth{Somphong Jitman and Todsapol Mankean}{Matrix-Product Constructions for Hermitian Self-Orthogonal Codes}

\title{Matrix-Product Constructions for Hermitian Self-Orthogonal Codes\thanks{This research was supported by the Thailand Research Fund under Research Grant MRG6080012.}}

\author{Somphong Jitman\thanks{S. Jitman (Corresponding Author) is with the Department of Mathematics, Faculty of Science,  Silpakorn University, Nakhon Pathom 73000,  Thailand.
        Email: sjitman@gmail.com} and Todsapol Mankean\thanks{T. Mankean is with the  Department of Mathematics, Faculty of Science,  Silpakorn University,  Nakhon Pathom 73000,  Thailand. Email: tong.todsapol2@gmail.com}}

\maketitle

\begin{abstract}
Self-orthogonal codes have been of interest due to there rich algebraic structures and wide applications.    Euclidean self-orthogonal  codes have been quite well studied in literature. 
Here,  we have  focused on Hermitian self-orthogonal codes.   Constructions of such codes have been given   based on the well-known  matrix-product construction for linear codes. Criterion for the  underlying matrix and the  input codes required in such constructions have been determined.  In many cases,    the Hermitian self-orthogonality of the input codes and the  assumption that the underlying matrix is unitary  can be  relaxed.  Some special matrices used in the constructions and illustrative  examples of good Hermitian self-orthogonal codes have been provided  as well.
\end{abstract}



\section{Introduction}


Self-orthogonal codes  constitute an important class of linear codes due to their rich algebraic structures and  wide  applications (see  \cite{JLLX2010},  \cite{JX2012}, \cite{P1972},  \cite{ZG2015}, and references therein).   In \cite{BN2001},  a nice construction   that can produce linear codes with explicit and good parameters has been introduced, namely, a matrix-product construction.     The said construction can be viewed as    a generalization of the well-known  
$(u|u + v)$-construction and  $(u + v + w|2u + v|u)$-construction  (see  \cite{BN2001}).     In   \cite{BN2001},  properties of matrix-product codes have been studied as well as a lower bound for the minimum distance of the output codes.  In some cases, the lower bound  given in \cite{BN2001}  has been  shown to be sharped  in \cite{HLR2009}. 

In \cite{GHR2015},  the  matrix-product  construction has been applied in constructing   Euclidean self-orthogonal codes in the case where  the underlying matrix is   a square  orthogonal matrix.   In the same fashion, this  idea  has been extended to construct  Hermitian  self-orthogonal codes in \cite{LLY2016} and \cite{ZG2015}.  However, in both cases, the input codes are required to be   self-orthogonal and the underlying matrix must be either orthogonal or unitary. 
For the Euclidean case, the Euclidean self-orthogonality of the input codes and the assumption that the matrix is orthogonal have  been relaxed in  \cite{MJ2017}.

In this paper, we  extend the concept in \cite{MJ2017} to cover  the  Hermitian case. The Hermitian self-orthogonality of the input codes and the assumption that the underlying is unitary can be  relaxed in many cases. Matrices used in the constructions are studied together with examples of some good  Hermitian self-orthogonal matrix-product codes.

The  paper is organized as follows. 
Some basic properties of matrices, linear codes, self-orthogonal codes, and matrix-product codes are recalled in Section 2. Two matrix-product constructions  for Hermitian self-orthogonal codes are discussed in Section 3. In Section 4, the study of   special matrices over finite fields  is  given. Illustrative examples of good matrix-product Hermitian self-orthogonal codes  are provided in Section 5.

\section{Preliminaries}
Let $q$ be a prime power and let $\mathbb{F}_q$ denote the finite field of order $q$. Some properties of matrices  and codes over $\mathbb{F}_q$ used in this paper are recalled  in the following subsections.

\subsection{Matrices}
For positive integers $s\leq l$, denote by $M_{s,l}(\mathbb{F}_q)$ the set of $s\times l$ matrices whose entries are from $\mathbb{F}_q$. 
A matrix $ A\in M_{s,l}(\mathbb{F}_q) $ is said to be {\em full-row-rank} if the rows of $A$ are linearly independent.  
Denote by $\mathrm{diag}(\lambda_1,\lambda_2,\dots, \lambda_s)$ the $s\times s$ {\em diagonal matrix} whose diagonal entries are  $\lambda_1,\lambda_2,\dots, \lambda_s$.  Similarly, let  $\mathrm{adiag}(\lambda_1,\lambda_2,\dots, \lambda_s)$ denote the  $s\times s$ {\em anti-diagonal matrix} whose anti-diagonal entries are  $\lambda_1,\lambda_2,\dots, \lambda_s$.   Denote by $I_s$ and $J_s$ the matrices $\mathrm{diag}(1,1,\dots,1)$ and $\mathrm{adiag}(1,1,\dots,1)$, respectively.

Assume that $q=r^2$ is square.  For a matrix $A=[a_{ij}]\in M_{s,l}(\mathbb{F}_q)$, let $A^\dagger:= [ \overline{a_{ji}}]\in M_{s,l}(\mathbb{F}_q)$, where $\overline{a}:=a^r$ for all $a \in \mathbb{F}_q$. In  this paper,  a     matrix $A\in M_{s,l}(\mathbb{F}_q)$ with the property that $AA^{\dag}$ is diagonal or  anti-diagonal   is required in  the  constructions of Hermitian self-orthogonal codes. To the best of our knowledge, there are no proper names  for such matrices.   For convenience,  the following definitions are given. A matrix $A \in M_{s,l}(\mathbb{F}_q)$ is said to be  {\em weakly semi-unitary} if $AA^{\dag}$ is diagonal and it is said to be  {\em weakly anti-semi-unitary} if $AA^{\dag}$ is anti-diagonal. In the case where $A$ is square, such matrices are called   {\em weakly quasi-unitary} and  {\em weakly anti-quasi-unitary}, respectively.
A matrix $A\in M_{s,s}(\mathbb{F}_q) $ is called a   {\em unitary  matrix} if $AA^\dagger=I_s$ and it  is called a   {\em quasi-unitary matrix} if $AA^\dagger=\lambda I_s$ for some non-zero $\lambda\in \mathbb{F}_q$.

\subsection{Linear Codes}
For each positive integer $n$, denote by $\mathbb{F}_q^n$ the $\mathbb{F}_q$-vector space of all vectors of length $n$ over $\mathbb{F}_q$.    A set $C\subseteq \mathbb{F}_q^n$ is called  a {\em linear code of length} $n$ over $\mathbb{F}_q$  if it is a subspace of the vector space  $\mathbb{F}_q^n$. A linear code $C$  of length $n$ over $\mathbb{F}_q$ is said to have parameters   $[n,k,d]_q$ if the  $\mathbb{F}_q$-dimension  of $C$ is $k$ and   the {\em minimum Hamming weight}  $d(C)$ of $C$ is 
\[{d}:=\min\{{\mathrm wt}(\boldsymbol{u})\mid \boldsymbol{u}  \in C\setminus \{0\}  \},\] 
where ${\mathrm wt}(\boldsymbol{u})$ is the number of nonzero entries in $\boldsymbol{u}$.
A  $k\times n$ matrix $G$ over $\mathbb{F}_q$ is called a {\em generator matrix} for an $[n,k,d]_q$ code $C$ if the rows of $G$  form a basis of $C$.

In addition,  assume that $q=r^2$ is a  square.   For  $\boldsymbol{v}=(v_1,v_2,\dots,v_n)\in \mathbb{F}_q^n$, let   $\overline{\boldsymbol{v}}:=(\overline{v_1},\overline{v_2},\dots,\overline{v_n})$, where $\overline{a}:=a^r$ for all $a\in \mathbb{F}_q$.  
For  $\boldsymbol{u}=(u_1,u_2,\dots,u_n)$ and  $\boldsymbol{v}=(v_1,v_2,\dots,v_n)$ in $\mathbb{F}_q^n$,  the {\it Hermitian inner product} of $\boldsymbol{u}$ and
$\boldsymbol{v}$ is defined by
\[\langle \boldsymbol{u},\boldsymbol{v}\rangle_H=\sum_{i=1}^nu_i\overline{v_i}.\] The {\em Hermitian dual}  of a linear code $C\subseteq \mathbb{F}_q^n$ is defined to be the set
\[C^{\perp_H}=\{\boldsymbol{v} \in \mathbb{F}_q^n\mid \langle \boldsymbol{u},\boldsymbol{v}\rangle _H=0 \text{ for all } \boldsymbol{v}\in C\}.\]
A linear code $C$ is said to be {\it Hermitian self-orthogonal} (resp., {\it self-dual})  
if   $C\subseteq C^{\perp_H}$ (resp., $C=C^{\perp_H})$.  

For linear codes $C_1$ and $C_2$ of the same length over $\mathbb{F}_q$, it is well known that if $C_i$ is generated by $G_i$ for $i\in \{1,2\}$, then  $G_1G_2^\dagger=[\boldsymbol{0}]$ if and only if  $C_1\subseteq C_2^{\perp_H} $. Especially, $G_1G_1^\dagger=[\boldsymbol{0}]$  if and only if  $C_1$ is  Hermitian self-orthogonal.

\subsection{Matrix-Product Codes} A matrix-product  construction for linear codes has been introduced in \cite{BN2001} and extensively studied in  \cite{BS2015} and \cite{HLR2009}.  The major results are summarized as follows.   For each integers $1\leq s\leq l$, let $C_{i}$ be a  linear  $[m,k_i,d_i]_q$ code   over $\mathbb{F}_{q}$ with generator matrix $G_i$ and let  $A=[a_{ij}]\in M_{s,l}(\mathbb{F}_{q})$. The {\em matrix-product code} $[C_{1},C_{2},\cdots,C_{s}]\cdot A$ is defined to be  the linear code of length $ml$ over $\mathbb{F}_{q}$  with generator matrix

\[G=\left[
\begin{array}{cccc}
a_{11}G_{1} & a_{12}G_{1} & \cdots & a_{1l}G_{1} \\
a_{21}G_{2} & a_{22}G_{2} & \cdots & a_{2l}G_{2} \\
\vdots & \vdots & \ddots & \vdots \\
a_{s1}G_{s} & a_{s2}G_{s} & \cdots & a_{sl}G_{s} \\
\end{array}
\right].\] 
The    matrix-product code $[C_{1},C_{2},\cdots,C_{s}]\cdot A$ is simply denoted by $C_A$ if $C_1,C_2,\dots, C_s$ are clear in the context.

For each  $A\in M_{s,l}(\mathbb{F}_{q})$ and for each $1\leq i\leq s$, denote by 
$\delta_{i}(A)$  the minimum weight  of the  linear code of length $l$ over $\mathbb{F}_{q}$ generated by the first $i$ rows of   $A$.  Some properties of  matrix-product codes are given in the following theorem. 

\begin{theorem}\label{thm0} Assume the notations above. Then the following statements hold.
    \begin{enumerate}
        \item $C_A$  is a linear code of length $ml$ over $\mathbb{F}_q$.
        \item   $\dim(C_A)\leq  \sum\limits_{i=1}^{s}k_{i}$.
        \item If $A$ is full-row-rank, then  $\dim(C_ A)=\sum\limits_{i=1}^{s}k_{i}$.
        \item     ${\mathrm d_\mathrm{H}}(C_ A)\geq  \min\limits_{1\leq i\leq s}\{d_{i}\delta_{i}(A)\}$.  
        \item    If  $C_1\supseteq C_2\supseteq \dots \supseteq C_s$, then  ${\mathrm d}(C_ A)= \min\limits_{1\leq i\leq s}\{d_{i}\delta_{i}(A)\}.$ 
    \end{enumerate}
\end{theorem}

From now on, we assume that $q=r^2$ is a square and focus on the dual   of a matrix product code with respect to the Hermitian inner product. 
If $A$ is an invertible square matrix, the Hermitian  dual of  a matrix-product   codes is again matrix-product and determined as follows.

\begin{theorem} \label{thm21}
    Assume the notation above and $s=\ell$. If $A$ is a nonsingular $s\times s$ matrix, then 
    $$([C_{1},C_{2},\cdots,C_{s}]\cdot A)^{\bot_H}=[C_{1}^{\bot_H},C_{2}^{\bot_H},\cdots,C_{s}^{\bot_H}]\cdot (A^{-1})^\dagger.$$ 
\end{theorem}
From Theorem \ref{thm21},     a matrix-product  construction  has been applied for   Hermitian self-orthogonal codes in   \cite{LLY2016} and \cite{ZG2015}, where   $A$ is a $s\times s$ unitary matrix    and the input codes $C_i$' are Hermitian  self-orthogonal. 

In general the Hermitian dual of a matrix-product code does not need to be matrix-product. 
In this paper, we focus on  this general set up for Hermitian self-orthogonal matrix-product codes. The restriction on the Hermitian self-orthogonality of the input codes and the condition that $A$ is unitary can be  relaxed in many cases.    The detailed constructions are given in the  next section.

\section{Constructions}	
 In this section,  we focus on two  types of  matrix-product constructions for Hermitian self-orthogonal linear codes.    Sufficient conditions on the matrix and the  input codes for matrix-product codes to be Hermitian self-orthogonal are given.

In the following theorem, a matrix-product construction for Hermitian  self-orthogonal codes whose input codes are Hermitian self-orthogonal is discussed. The results    are a bit more general than the ones in \cite{GHR2015} since the underlying matrix does not need to be unitary. The construction is given as follows.
\begin{theorem} \label{thm3}
    Let    $s \leq l$ be positive integers. Let  $C_{1},C_{2},\dots ,C_{s}$  be linear codes of the same length over $\mathbb{F}_q$ and let  $A  \in M_{s \times l}(\mathbb{F}_q)$.
    If $AA^\dag $ is diagonal and $C_{i} \subseteq C_{i}^{{\perp}_{H}}$ for all $1\leq i \leq s$, then $C_{A} \subseteq C_{A}^{{\perp}_{H}}$.
\end{theorem}

\begin{proof}
    Assume that $AA^\dag = \mathrm{diag}(\lambda _{1},\lambda _{2},\dots,\lambda _{s})$ and $ \ C_{i} \subseteq \ C_{i}^{{\perp}_{H}} \ $ for all  $1\leq i \leq s$.  For each $1\leq i\leq s$,  let $G_i$ be a generator matrix for the code  $C_i$.  Assume that  $A = \begin{bmatrix}
    a_{11} &a_{12}  &\cdots   &a_{1l} \\ 
    a_{21} & a_{22} & \cdots  &a_{2l} \\ 
    \vdots  &  \vdots & \ddots  &\vdots  \\ 
    a_{s1}&a_{s2}  &  \cdots & a_{sl}
    \end{bmatrix}$. Then the matrix-product code $C_A$ is generated by
    \[ G = \begin{bmatrix}
    a_{11}G_{1} &a_{12}G_{1}  &\cdots   &a_{1l}G_{1} \\ 
    a_{21}G_{2} & a_{22}G_{2} & \cdots  &a_{2l}G_{2} \\ 
    \vdots  &  \vdots & \ddots  &\vdots  \\ 
    a_{s1}G_{s}&a_{s2}G_{s}  &  \cdots & a_{sl}G_{s}
    \end{bmatrix}.\]
    It follows that 
    \begin{align*}
    GG^\dag & = \begin{bmatrix}
    a_{11}G_{1} &a_{12}G_{1}  &\cdots   &a_{1l}G_{1} \\ 
    a_{21}G_{2} & a_{22}G_{2} & \cdots  &a_{2l}G_{2} \\ 
    \vdots  &  \vdots & \ddots  &\vdots  \\ 
    a_{s1}G_{s}&a_{s2}G_{s}  &  \cdots & a_{sl}G_{s}
    \end{bmatrix}\begin{bmatrix}
    a_{11}^{r}G_{1}^\dag &a_{21}^{r}G_{2}^\dag  &\cdots   &a_{s1}^{r}G_{s}^\dag \\ 
    a_{12}^{r}G_{1}^\dag & a_{22}^{r}G_{2}^\dag & \cdots  &a_{s2}^{r}G_{s}^\dag \\ 
    \vdots  &  \vdots & \ddots  &\vdots  \\ 
    a_{1l}^{r}G_{1}^\dag&a_{2l}^{r}G_{2}^\dag  &  \cdots & a_{sl}^{r}G_{s}^\dag
    \end{bmatrix}.\\ 
    & =   \begin{bmatrix}
    \lambda_{1}(G_{1}G_{1}^\dag) & 0(G_{1}G_{2}^\dag) & \cdots & 0(G_{1}G_{s}^\dag)\\ 
    0(G_{2}G_{1}^\dag) & \lambda_{2}(G_{2}G_{2}^\dag) & \cdots  & 0(G_{2}G_{s}^\dag)\\ 
    \vdots  &  \vdots & \ddots  &\vdots  \\  
    0(G_{s}G_{1}^\dag) & 0(G_{s}G_{2}^\dag) &  \cdots & \lambda_{s}(G_{s}G_{s}^\dag)
    \end{bmatrix} .
    \end{align*}
    Since $ \ C_{i} \subseteq \ C_{i}^{{\perp}_{H}} \ $ for all $1\leq i \leq s$,  we have $G_{i}G_{i}^\dag = [\boldsymbol{0}]$  for all $1\leq i \leq s$. It follows that 
    $GG^\dag = [\boldsymbol{0}]$.  Hence,  $C_{A} \subseteq C_{A}^{{\perp}_{H}}$ as desired.
\end{proof}
If $A$ is a square quasi-unitary, then the following corollary can be deduced.

\begin{corollary}	
    If  $A\in M_{s,s}(\mathbb{F}_q)$  is such that  $AA^\dag = \lambda I_{s}$ for some non-zero $\lambda$ in $\mathbb{F}_q$ and $C_{i} \subseteq C_{i}^{{\perp}_{H}}$ for all $1\leq i \leq s$, then $C_{A} \subseteq C_{A}^{{\perp}_{H}}$.
\end{corollary}	

\begin{example} Let $\beta$ be a primitive element of $\mathbb{F}_4$ and let $A = 
    \begin{bmatrix} 
    1 & 1&1 \\
    1 & \beta & \beta^2  \\
    1 & \beta^2 & \beta \\
    \end{bmatrix}$. Then $A$ is invertible,   $AA^\dag = \mathrm{diag}(1,1,1)$, $\delta_1(A)=3, \delta_2(A)=2$ and $\delta_3(A)=1$ .  Let $C_1, C_2$ and $C_3$ be the linear codes of length $6$ over $\mathbb{F}_4$ generated by 
    \[G_{1} = 
    \begin{bmatrix} 
    1& 1 & 1 & 1& 1 & 1 \\
    1 & \beta & \beta^2& \beta^3 & \beta^4 & \beta^5\\
    1 & \beta^2 & \beta^4& \beta^6 & \beta^8 & \beta^{10}\\
    \end{bmatrix},~~~ G_2 =\begin{bmatrix}  	1& 1 & 1 & 1& 1 & 1 \\
    1 & \beta & \beta^2& \beta^3 & \beta^4 & \beta^5     
    \end{bmatrix} ,~~\text{
    and } 	~G_{3} = 
    \begin{bmatrix} 
    1& 1 & 1 & 1& 1 & 1 \\
    \end{bmatrix},\]
    respectively.  Then $C_1, C_2$ and $C_3$ are Hermitian self-orthogonal with parameters $[6, 3, 2]_4, [6,2,4]_4$ and  $[6,1,6]_4$ respectively. Since $C_3 \subseteq C_2\subseteq C_1$, 
    by Theorems \ref{thm0} and \ref{thm3}, $C_A$ is a Hermitian self-orthogonal code with parameters $[18,6,6]_4$.
\end{example}

In the following theorem,  a matrix-product construction for Hermitian self-orthogonal codes is studied in the case where  the Hermitian self-orthogonality of the input codes  are relaxed.

\begin{theorem}\label{thm4}
    Let $s \leq l$ be positive integers. Let  $C_{1},C_{2},\dots ,C_{s}$  be linear codes of the same length over $\mathbb{F}_q$ and let  $A  \in M_{s \times l}(\mathbb{F}_q)$.
    If $AA^\dag $ is anti-diagonal and  $C_{i} \subseteq C_{s-i+1}^{{\perp}_{H}}$  for all $1\leq i \leq s$, then $C_{A} \subseteq C_{A}^{{\perp}_{H}}$.
\end{theorem}

\begin{proof}
    Assume that $AA^\dag = \mathrm{adiag}(\lambda_{1},\lambda_{2},\dots,\lambda_{s})$ and $ \ C_{i} \subseteq  C_{s-i+1}^{{\perp}_{H}} $ for all  $1\leq i \leq s  $.
    For each $1\leq i\leq s$, let $G_i$ be a generator matrix of the code $C_i$.  Write 
    $A = \begin{bmatrix}
    a_{11} &a_{12}  &\cdots   &a_{1l} \\ 
    a_{21} & a_{22} & \cdots  &a_{2l} \\ 
    \vdots  &  \vdots & \ddots  &\vdots  \\ 
    a_{s1}&a_{s2}  &  \cdots & a_{sl}
    \end{bmatrix}$. Then the matrix-product code $C_A$ is generated by 
    \[G = \begin{bmatrix}
    a_{11}G_{1} &a_{12}G_{1}  &\cdots   &a_{1l}G_{1} \\ 
    a_{21}G_{2} & a_{22}G_{2} & \cdots  &a_{2l}G_{2} \\ 
    \vdots  &  \vdots & \ddots  &\vdots  \\ 
    a_{s1}G_{s}&a_{s2}G_{s}  &  \cdots & a_{sl}G_{s}
    \end{bmatrix}.\]
    It follows that 
    \begin{align*}
    GG^\dag & = \begin{bmatrix}
    a_{11}G_{1} &a_{12}G_{1}  &\cdots   &a_{1l}G_{1} \\ 
    a_{21}G_{2} & a_{22}G_{2} & \cdots  &a_{2l}G_{2} \\ 
    \vdots  &  \vdots & \ddots  &\vdots  \\ 
    a_{s1}G_{s}&a_{s2}G_{s}  &  \cdots & a_{sl}G_{s}
    \end{bmatrix}\begin{bmatrix}
    a_{11}^{r}G_{1}^\dag &a_{21}^{r}G_{2}^\dag  &\cdots   &a_{s1}^{r}G_{s}^\dag \\ 
    a_{12}^{r}G_{1}^\dag & a_{22}^{r}G_{2}^\dag & \cdots  &a_{s2}^{r}G_{s}^\dag \\ 
    \vdots  &  \vdots & \ddots  &\vdots  \\ 
    a_{1l}^{r}G_{1}^\dag&a_{2l}^{r}G_{2}^\dag  &  \cdots & a_{sl}^{r}G_{s}^\dag
    \end{bmatrix}\\ 
    & =   \begin{bmatrix}
    0(G_{1}G_{1}^\dag)  & \cdots & 0(G_{1}G_{s-1}^\dag)  & \lambda_{1}(G_{1}G_{s}^\dag)\\ 
    0(G_{2}G_{1}^\dag)   & \cdots  &\lambda_{2}(G_{2}G_{s-1}^\dag)& 0(G_{2}G_{s}^\dag)\\ 
    \vdots  &    \iddots & \vdots  &\vdots  \\  
    \lambda_{s}(G_{s}G_{1}^\dag)   &  \cdots & 0(G_{s}G_{s-1}^\dag) & 0(G_{s}G_{s}^\dag)
    \end{bmatrix} .
    \end{align*}
    Since $ \ C_{i} \subseteq C_{s-i+1}^{{\perp}_{H}}  $ for all $1\leq i \leq s  $,  we have  $G_{i}G_{s-i+1}^\dag = [\boldsymbol{0}]$  for all $1\leq i \leq s$. 
    Hence, $GG^\dag = [\boldsymbol{0}]$. Therefore,   $C_{A} \subseteq C_{A}^{{\perp}_{H}}$ as desired.
\end{proof}
The following results can be deduced directly from Theorem \ref{thm4}. The proofs are omitted.
\begin{corollary}	 \label{cor415}
    If $A \in M_{s,s}(\mathbb{F}_q)$ is such that $AA^\dag = \lambda J_{s}$ for some non-zero $\lambda$ in $\mathbb{F}_q$ and $C_{i} \subseteq C_{s-i+1}^{{\perp}_{H}}$ for all $1\leq i \leq s$, then $C_{A} \subseteq C_{A}^{{\perp}_{H}}$.	
\end{corollary}	
\begin{example} Let $\beta$ be a primitive element of $\mathbb{F}_4$. Then $A = 
    \begin{bmatrix} 
    1  & \beta \\
    \beta & 1
    \end{bmatrix}$ is  invertible,  $AA^\dag = \mathrm{adiag}(1,1)$, $\delta_1(A)=2$, and $\delta_2(A)=1$.  Let $C_1$ and $C_2$ be the linear codes of length $4$ over $\mathbb{F}_4$ generated by \[G_{1} = 
    \begin{bmatrix} 
    1& 1 & 1 & 1\\
    0 & 0 & \beta & \beta \\
    \end{bmatrix} ~\text{ and }
    ~G_2 =\begin{bmatrix}         1 & 1& 1&  1 
    \end{bmatrix} ,\]     respectively.  Then $C_1$ and $C_2$ have  parameters $[4,2,2]_4$ and  $[4,1,4]_4$, respectively. Since $C_2\subseteq C_1\subseteq C_2^{{\perp}_{H}}$, 
    by Theorem \ref{thm0} and Corollary \ref{cor415}, $C_A$ is a Hermitian self-orthogonal code with parameters $[8,3,4]_4$.
\end{example}
By choosing $C_{i} = {C}^{{\perp}_{H}}_{s-i+1}$ in Corollary \ref{cor415}, we have the following result.

\begin{corollary}	\label{cor3.3}
    If $A \in M_{s,s}(\mathbb{F}_q)$ is  such that  $AA^\dag = \lambda J_{s}$  for some non-zero $\lambda$ in $\mathbb{F}_q$ and $C_{i} = C_{s-i+1}^{{\perp}_{H}}$ for all $1\leq i\leq s$, then $C_{A}$ is Hermitian self-dual.	
\end{corollary}	
\section{Special Matrices and Applications}	

As discussed  in Section 3,     a   full-row-rank  matrix $A\in M_{s,l}(\mathbb{F}_q)$ with the property that $AA^{\dag}$ is diagonal or  anti-diagonal  is required in the  matrix-product construction for Hermitian self-orthogonal codes.   From Theorem \ref{thm0},  the minimum Hamming weight of the output code  depends on the sequence   $\{\delta_i(A) \}_{i=1,2,\dots,s}$. In most cases,  the output code has large  minimum Hamming weight if the sequence   $\{\delta_i(A) \}_{i=1,2,\dots,s}$ is decreasing.

In this section,  a certification for  the existence of   {weakly quasi-unitary} and  {weakly anti-quasi-unitary}  matrices $A$ over some finite fields with the property that  $\{\delta_i(A) \}_{i=1,2,\dots,s}$ is a decreasing sequence.  Precise applications of such matrices in constructing Hermitian self-orthogonal codes are explained.

\subsection{Weakly Quasi-Unitary Matrices}  In this subsection, some weakly quasi-unitary matrices with   the property that the sequence $\{\delta_i(A) \}_{i=1,2,\dots,s}$ is decreasing are  given as well as their applications in a matrix-product construction of Hermitian self-orthogonal codes.

First, we consider  $2\times 2$ (weakly) quasi-unitary  matrices over  an arbitrary finite field of square order greater than $4$. 
\begin{lemma}\label{lem421}
    Let $r$ be  a prime power and $q=r^2$. Let $ \alpha$ be a primitive element of $\mathbb{F}_{q}$.  Then one of  the following statements holds.
    \begin{enumerate}
        \item If $q$ is odd, then $A=
        \left[\begin{array}{cc}
        1 & 1\\
        1 & -1
        \end{array}\right] \in M_{2,2}(\mathbb{F}_q)$ is   invertible and  (weakly) quasi-unitary  with $\delta_1(A)=2$ and $\delta_2(A)=1$.
        \item If $q \geq 16$ is even, then $A = \begin{bmatrix}
        1  & \alpha \\ 
        \alpha^r  & 1 
        \end{bmatrix}  \in M_{2,2}(\mathbb{F}_q)$	is invertible and (weakly) quasi-unitary with $\delta_1(A)=2$ and  $\delta_2(A)=1$.  
    \end{enumerate}
\end{lemma}	

\begin{proof}
    To prove $(1)$, assume that $q$ is odd. Let $A=
    \left[\begin{array}{cc}
    1 & 1\\
    1 & -1
    \end{array}\right]$. Clearly, $A$ is invertible, $\delta_{1}(A) = 2$ and $\delta_{2}(A) = 1$. Since
    $$AA^{T} = \begin{bmatrix}
    1 & 1\\ 
    1 & -1  
    \end{bmatrix} \begin{bmatrix}
    1 & 1  \\ 
    1 & -1  
    \end{bmatrix}  
    = \begin{bmatrix}
    2 & 0  \\ 
    0 & 2  
    \end{bmatrix} = \mathrm{diag}(2,2),$$
    $A$ is (weakly) quasi-unitary.
    
   To prove  $(2)$,  assume that $q > 2$ is even. Let $A = \begin{bmatrix}
    1  & \alpha \\ 
    \alpha^r  & 1 
    \end{bmatrix}$. Clearly, $A$ is invertible, $\delta_1(A)=2$ and  $\delta_2(A)=1$.
    Since 
    \begin{align*}
    AA^\dag & = \begin{bmatrix}
    1  & \alpha \\ 
    \alpha^r  & 1 
    \end{bmatrix} \begin{bmatrix}
    1  & \alpha^{r^2} \\ 
    \alpha^r  & 1 
    \end{bmatrix} \\
    & = \begin{bmatrix}
    1+\alpha^{r+1} & \alpha^{r^2}+\alpha \\ 
    \alpha^r +  \alpha^r & 1+\alpha^{r+1} 
    \end{bmatrix}\\
    & = \begin{bmatrix}
    1+\alpha^{r+1}  & 0 \\ 
    0  & 1+\alpha^{r+1} 
    \end{bmatrix}\\ 
    & = \mathrm{diag}(1+\alpha^{r+1},1+\alpha^{r+1}),
    \end{align*}
    A is (weakly) quasi-unitary.
\end{proof}	

\begin{remark} We note that for every  $2\times 2$ matrix $A=\begin{bmatrix}
    a  & b \\ 
    c  & d
    \end{bmatrix}$ over $\mathbb{F}_4$,  if $\delta_1(A)=2$, then $a$ and $b$ are non-zeros.    Hence, the top-left conner of $AA^\dagger$ is $a^3+b^3=1+1=0$.    Hence, $A$ cannot be   weakly quasi-unitary. Therefore, there are no   weakly  quasi-unitary matrices in $M_{2,2}(\mathbb{F}_4)$ with $\delta_1(A)=2$ and $\delta_2(A)=1$.

\end{remark}
Quasi-unitary matrices in Lemma \ref{lem421} can be applied to construct Hermitian self-orthogonal codes as follows.
\begin{corollary}\label{cor09} Let $r\geq 3$  be a prime power and let $q=r ^2$.   If there exist Hermitian self-orthogonal $[m,k_1,d_1]_q$ and  $[m,k_2,d_2]_q$ codes, then a Hermitian self-orthogonal $[2m,k_1+k_2, d]_q$ code can be constructed with $d\geq \min\{2d_1,d_2\}$.
\end{corollary}
\begin{proof}
    Assume that there exist Hermitian self-orthogonal codes $C_1$ and $C_2$ with parameters  $[m,k_1,d_1]_q$ and  $[m,k_2,d_2]_q$. By Lemma \ref{lem421}, there exist  a $2\times 2$   invertible and (weakly) quasi-unitary  matrix $A$ over $\mathbb{F}_q$ with $\delta_1(A)=2$ and $\delta_2(A)=1$. By Theorems \ref{thm0} and \ref{thm3}, the matrix-product code $C_{A}$ is Hermitian self-orthogonal  with parameters $[2m,k_1+k_2, d]_q$   with $d\geq \min\{2d_1,d_2\}$. 
\end{proof} 
\begin{example} Let $\beta$ be a primitive element of $\mathbb{F}_{16}$. By Lemma \ref{lem421}, we have that $A = 
    \begin{bmatrix} 
    1  & \beta\\ 
    \beta^2  & 1 
    \end{bmatrix}$ is  invertible,  {$AA^\dag = \mathrm{diag}(1+ \beta^{5},1+ \beta^{5})$, } $\delta_1(A)=2$ and $ \delta_2(A)=1$.  Let $C_1$ and $ C_2$  be the linear codes of length $4$ over $\mathbb{F}_{16}$ generated by 
    \[G_1 =\begin{bmatrix}  	1& 1 & 1 & 1 \\
    1 & 0 & 1 & 0   
    \end{bmatrix} ,\]  
    and 	\[G_{2} = 
    \begin{bmatrix} 
    1& 1 & 1 & 1\\
    \end{bmatrix},\]
    respectively.  Then $C_1$ and $ C_2$ are Hermitian self-orthogonal with parameters $[4, 2, 2]_{16}$ and $ [4,2,1]_{16}$  respectively. Since $C_2\subseteq C_1$, 
    by Theorem \ref{thm0} and Corollary \ref{cor09}, $C_A$ is a Hermitian self-orthogonal code with parameters $[8,3,4]_{16}$.
\end{example}

Next,  we focus on $M \times M$  (weakly) quasi-unitary matrices  over $\mathbb{F}_{q}$,  where $M\geq 2$ is an integer. 

\begin{lemma}	\label{Lem44} Let $r$ be a prime power and $q=r^2$.
    Let $M$ be a positive integer.  If $M | (r+1)$, then there exists an $M \times M$  (weakly) quasi-unitary  matrix over $\mathbb{F}_{q}$ with $\delta_i(A)=M-i+1$ for all $1 \leq i \leq M$.  
\end{lemma}
\begin{proof} 
    Assume that $M | (r+1)$. Then $\mathbb{F}_{q}$ contains a primitive $M$-th root unity. Let $\alpha$ be a fixed primitive $M$-th root unity in $\mathbb{F}_{q}$. Since $r \equiv  -1 \,\mathrm{mod}\, M$, we have
    $$ \overline{\alpha} = \alpha^{r} =\alpha^{-1}.$$
    Define\\
    $$ A =\begin{bmatrix}
    (\alpha^0)^{0} & (\alpha^1)^{0} & \cdots  & (\alpha^{M-1})^{0} \\ 
    (\alpha^0)^{1}& (\alpha^1)^{1} & \cdots  & (\alpha^{M-1})^{1} \\ 
    \vdots & \vdots  &  \ddots  & \vdots  \\ 
    (\alpha^0)^{M-1}&  (\alpha^1)^{M-1}& \cdots   & (\alpha^{M-1})^{M-1} 
    \end{bmatrix}.$$
    Let $B = AA^\dag$. Then, for all $1 \leq i,j \leq M$, we have
    \begin{align*}
    b_{ij} & = \sum_{k=0}^{M-1}(\alpha^{k})^{i-1} \overline{(\alpha^{k})^{j-1}} =  \sum_{k=0}^{M-1}(\alpha^{k})^{i-1} \overline{(\alpha^{k})}^{j-1}\\
    & = \sum_{k=0}^{M-1}(\alpha^{k})^{i-1}(\alpha^{-k})^{j-1} = \sum_{k=0}^{M-1}(\alpha^{i-j})^{k}\\
    & =  \left\{\begin{matrix}
    M\neq 0 & \text {if $i = j$}, \\ 
    0 & \ \ \ \ \ \text{if otherwise.}
    \end{matrix}\right.
    \end{align*}
    Hence, $AA^\dag = \mathrm{diag}(M,M,\dots, M)=MI_M$. Therefore, A is (weakly) quasi-unitary. From {\cite[Theorem 3.2]{BN2001}},  it follows that  $\delta_i(A)=M-i+1$ for all $1 \leq i \leq M$. 
\end{proof}	
\begin{corollary}\label{cor011}  Let $r$ be a prime power and $q=r^2$. Let  $M$ be a  positive integer such that $M | (r+1)$.  If there exist Hermitian self-orthogonal $[m,k_1,d_1]_q$, $[m,k_2,d_2]_q , \dots,[m,k_M,d_M]_q$ codes, then a Hermitian self-orthogonal $[Mm,k_1+k_2+\cdots+k_M, d]_q$ code can be constructed with $d\geq \min\{Md_1,(M-1)d_2, \dots,d_M\}$.
\end{corollary}

\begin{proof} \
    Assume that there are $M$ Hermitian self-orthogonal   codes with parameters $[m,k_1,d_1]_q$,  $[m,k_2,d_2]_q, \dots, [m,k_M,d_M]_q$. By Lemma \ref{Lem44}, there exist an  $M \times M$    invertible and quasi-unitary matrix $A$ over $\mathbb{F}_q$ with $\delta_1(A)=M, \delta_2(A)=(M-1), \dots ,\delta_M(A) =1$.  By Theorems \ref{thm0} and \ref{thm3}, the  matrix-product code  $C_{A}$ is  Hermitian self-orthogonal with parameters $[Mm,k_1+k_2+\cdots+k_M, d]_q$  with $d\geq \min\{Md_1,(M-1)d_2, \dots, d_M\}$. 
\end{proof}

\begin{example}
    Let $\alpha$ be a primitive element of $\mathbb{F}_4$. Then $\alpha$ is primitive $3$th root unity in $\mathbb{F}_4$. By Lemma \ref{Lem44}, it follows  that $A =\begin{bmatrix}
    1 & 1 & 1   \\ 
    1 & \alpha & \alpha^2 \\ 
    1 & \alpha^2  &  \alpha^4
    \end{bmatrix}$ is  invertible,  $AA^\dag = \mathrm{diag}(1,1,1)$, $\delta_1(A)=3,  \delta_2(A)=2 $ and $\delta_1(A)=1$. Let $C_1$ ,$C_2$ and $C_3$ be the linear codes of length $6$ over $\mathbb{F}_4$ generated by \[G_{1} = 
    \begin{bmatrix} 
    1& 1 & 1 & 1& 1 & 1\\
    0 & 0& 1& 1& \alpha & \alpha \\
    0 & 0& 0& 0& 1 & 1\\
    \end{bmatrix},\]   \[G_{2} = 
    \begin{bmatrix} 
    1& 1 & 1 & 1& 1 & 1\\
    0 & 0& 1& 1& \alpha  & \alpha \\
    \end{bmatrix}\]
    and 
    \[G_3 =\begin{bmatrix}         1 & 1& 1&  1 & 1& 1
    \end{bmatrix} ,\]    respectively.  Then  $C_3\subseteq C_2\subseteq C_1$ are Hermitian self-orthogonal with parameters $[6,3,2]_4, [6,2,4]_{4}$ and  $[6,1,6]_4$, respectively.  By Theorems \ref{thm0}  and \ref{cor011} $C_A$ is a Hermitian self-orthogonal code with parameters $[18,6,6]_4$.
\end{example}

\subsection{Weakly Anti-Quasi-Unitary Matrices}
In this subsection, we focus on the existence of weakly anti-quasi-unitary matrices with   the property that the sequence $\{\delta_i(A)\}_{i=1,2,\dots,s}$ is decreasing.   Their applications in  constructing Hermitian self-orthogonal codes are discussed as well.

In a finite field $\mathbb{F}_{q}$  with  $q = r^2$, the norm function $N : \mathbb{F}_{q} \rightarrow \mathbb{F}_{r}$ is defined by $N(\alpha) = \alpha^{r+1}$ for all $\alpha $ in $\mathbb{F}_{q}$.  In \cite[p. 57]{LN1997}, it has been shown that $N$ is surjective. Hence,  the following lemma and corollaries can be deduced.

\begin{lemma} \label{lem9} Let $r$ be a prime power and $q=r^2$. 
    Let $\alpha$ be a primitive element of $\mathbb{F}_{q}$.  Then the following statements hold.
    \begin{enumerate}
        \item If $q$ is odd, then there exists $b \in \mathbb{F}_{q}$ such that  $b^{r+1} = -1$ and $A=
        \left[\begin{array}{cc}
        1 & b\\
        b & 1
        \end{array}\right]$ is   invertible and (weakly) anti-quasi-unitary with $\delta_1(A)=2$ and $\delta_2(A)=1$.
        \item If $q\geq 4$ is even, then $A = \begin{bmatrix}
        \alpha   & \alpha^r \\ 
        1  & 1 
        \end{bmatrix}$	is invertible and (weakly) anti-quasi-unitary with $\delta_1(A)=2$ and  $\delta_2(A)=1$.  
    \end{enumerate}		
\end{lemma}	
\begin{proof}
   To prove $(1)$, assume that $q$ is odd.  Since the norm is surjective, there exists $b \in \mathbb{F}_q$  such that $N(b)= b^{r+1} = -1$. Let $A = \begin{bmatrix}
    1 & b \\ 
    b & 1 
    \end{bmatrix}$. Clearly, $A$ is invertible, $\delta_1(A)=2$ and  $\delta_2(A)=1$. Since
    \begin{align*}
    AA^\dag & = \begin{bmatrix}
    1 & b \\ 
    b & 1 
    \end{bmatrix} \begin{bmatrix}
    b  & b^{r} \\ 
    b^{r}  & 1 
    \end{bmatrix} \\
    & = \begin{bmatrix}
    1+b^{r+1} & b +b^{r} \\ 
    b+b^r & 1+b^{r+1}
    \end{bmatrix}\\
    & = \begin{bmatrix}
    0 & b + b^{r} \\ 
    b + b^{r}  & 0 
    \end{bmatrix}\\ 
    & = \mathrm{adiag}(b + b^{r} ,b + b^{r} ),
    \end{align*}
    A is (weakly) anti-quasi-unitary.
    
    To prove $(2)$, assume that $q \geq 4$ is even.  Let $A = \begin{bmatrix}
    \alpha  & \alpha^r \\ 
    1  & 1 
    \end{bmatrix}$. Clearly, $A$ is invertible, $\delta_1(A)=2$ and  $\delta_2(A)=1$.
    Since 
    \begin{align*}
    AA^\dag & = \begin{bmatrix}
    \alpha  & \alpha^r \\ 
    1  & 1 
    \end{bmatrix} \begin{bmatrix}
    \alpha^{r}  & 1 \\ 
    \alpha^{r^2}  & 1 
    \end{bmatrix} \\
    & = \begin{bmatrix}
    \alpha^{r+1}+\alpha^{r+1} & \alpha^{r}+\alpha \\ 
    \alpha^r +  \alpha  & 1+1
    \end{bmatrix}\\
    & = \begin{bmatrix}
    0  & \alpha^{r}+\alpha \\ 
    \alpha^{r}+\alpha  & 0
    \end{bmatrix}\\ 
    & = \mathrm{adiag}(\alpha^{r}+\alpha ,\alpha^{r}+\alpha ),
    \end{align*}
    A is (weakly) anti-quasi-unitary.
\end{proof}	
\begin{corollary} \label{lasther}  Let $r$ be a prime power and $q=r^2$.    If there exist  codes $C_1$ and $C_2$ with parameters  $[m,k_1,d_1]_q$ and  $[m,k_2,d_2]_q$ such that $C_1\subseteq C_2^{\perp_{H}}$, then a Hermitian self-orthogonal $[2m,k_1+k_2, d]_q$ code can be constructed with $d\geq \min\{2d_1,d_2\}$.
\end{corollary}
\begin{proof}
    Assume that there exist linear codes $C_1$ and $C_2$ with parameters  $[m,k_1,d_1]_q$ and  $[m,k_2,d_2]_q$ such that $C_1\subseteq C_2^{\perp_{H}}$. By Lemma \ref{lem9}, there exist  a $2\times 2$    invertible and  anti-quasi-orthogonal   matrix $A$ over $\mathbb{F}_q$ with $\delta_1(A)=2$ and $\delta_2(A)=1$. By Theorems \ref{thm0} and \ref{thm4}, the matrix-product code $C_{A}$ is Hermitian self-orthogonal  with parameters $[2m,k_1+k_2, d]_q$   with $d\geq \min\{2d_1,d_2\}$. 
\end{proof} 
\begin{example} Let $\beta$ be a primitive element of $\mathbb{F}_4$. By Lemma \ref{lem9}, we have that $A = 
    \begin{bmatrix} 
    \beta & \beta^{2} \\
    1 & 1
    \end{bmatrix}\in M_{2,2}(\mathbb{F}_4)$ is  invertible,  $AA^\dag = \mathrm{adiag}(1,1)$, $\delta_1(A)=2$, and $\delta_2(A)=1$.  Let $C_1$ and $C_2$ be the linear codes of length $6$ over $\mathbb{F}_4$ generated by \[G_{1} = 
    \begin{bmatrix} 
    1& 1 & 1 & 1& 1 & 1\\
    \beta & \beta^2& \beta^3& \beta^3& \beta^4 & \beta^5
    \end{bmatrix}\] 
    and 
    \[G_2 =\begin{bmatrix}         1 & 1& 1&  1 & 1& 1
    \end{bmatrix} ,\]     respectively.  Then $C_1$ and $C_2$ have  parameters $[6,2,4]_4$ and  $[6,1,6]_4$, respectively. Since $C_2\subseteq C_1\subseteq C_2^{{\perp}_{H}}$, 
    by Theorems \ref{thm0} and \ref{thm4}, $C_A$ is a Hermitian self-orthogonal code with parameters $[12,3,6]_4$.
\end{example}
%
%
%
%
\section{Examples}	
 In this section,   examples of Hermitian self-orthogonal matrix-product codes with good parameters are given.

 Using  Corollary \ref{cor09} and Hermitian self-orthogonal codes form various sources, Hermitian self-orthogonal matrix product codes can be constructed.  Here,   Hermitian self-orthogonal codes given in  \cite{JLLX2010}  are applied  in Corollary \ref{cor09},  and hence,   Hermitian self-orthogonal matrix-product codes with good parameters are obtained.

  In  \cite[Theorem 2.6]{JLLX2010},  it has been shown that  there exists a  Hermitian self-orthogonal $[q+1,k,q-k+2]_q$  code for all $2 \leq k \leq \frac{r}{2}$, where $q=r^2$.
By setting $C_{1}$ and $C_{2}$ be  Hermitian self-orthogonal codes with parameter $[q+1,\left \lfloor \frac{r}{2} \right \rfloor,q-\left \lfloor \frac{r}{2} \right \rfloor+2]_q$ and $[q+1,\left \lfloor \frac{r}{2} \right \rfloor-1,q -\left \lfloor \frac{r}{2} \right \rfloor +3 ]_q$ in  Corollary \ref{cor09}, we have the following result.

\begin{corollary}\label{cor431} Let $r$ be a prime power and $q = r^2$. Then  a Hermitian self-orthogonal $[2(q+1),2\left \lfloor \frac{r}{2} \right \rfloor -1 , d]_q$ code can be constructed with $d\geq q -\left \lfloor \frac{r}{2} \right \rfloor +3$.
\end{corollary}

From  Corollary \ref{cor431}, some examples of Hermitian self-orthogonal matrix-product codes over $\mathbb{F}_q$   are given in Table~\ref{Tab:SRNRValues}.

\renewcommand{\multirowsetup}{\centering} 
\setlength{\tabcolsep}{10pt}
\begin{table}[!hbt]
    \caption{Hermitian Self-Orthogonal Matrix-Product Codes over $\mathbb{F}_q$}
    \label{Tab:SRNRValues}
    \begin{center}
        \begin{tabular}{|c|c|c|c|}
            \hline
            \multirow{2}{1.5cm}{$q$}& \multicolumn{3}{p{5cm}|}{\centering Parameters} \bigstrut \\
            \cline{2-4} & \multicolumn{1}{c|}{$C_{1}$} & \multicolumn{1}{c|}{$C_{2}$} & \multicolumn{1}{c|}{$C_{A}$} \bigstrut \\ \hline
            49 & $[50,3,48]_{49}$ & $[50,2,49]_{49}$  & $[100,5,d]_{49}$ with $d \geq 49$   \bigstrut \\
            64 & $[65,4,62]_{64}$  & $[65,3,63]_{64}$   & $[130,7,d]_{64}$ with $d \geq 63$  \bigstrut \\
            81 & $[82,4,79]_{81}$  & $[82,3,80]_{81}$  & $[164,7,d]_{81}$ with $d \geq 80$   \bigstrut \\
            121 & $[122,5,118]_{121}$ & $[122,4,119]_{121}$ & $[244,9,d]_{121}$ with $d \geq 119$  \bigstrut \\
            \hline
        \end{tabular}
    \end{center}
\end{table}

Next, we focus on examples of Hermitian self-orthogonal matrix-product codes  derived from Corollary \ref{lasther}. For this case, good input codes can be chosen  from the family of Generalized Reed-Solomon (GRS) codes  recalled as follows.

For each positive integer $n\leq q$,
let $\gamma=(\gamma_{1},\gamma_2,\ldots,\gamma_{n})$ and $w=(w_{1},w_2,\ldots,w_{n})$ where $\gamma_i$ is a non-zero element
and $w_{1},w_2,\ldots w_n$ are distinct elements in $\mathbb{F}_q$.
For each $0\leq k\leq n$, denote by $\mathbb{F}_q[X]_{k}$ the set of all polynomials of degree less
than $k$ over $\mathbb{F}_q$. For convenience, the degree of the zero polynomial is defined to be $-1$.
A {\em GRS code} of length $n \leq q$ and dimension $k \leq n$ is defined to be 
\begin{align}\label{def:GRS}
{GRS}_{n,k}(\gamma,w):=
\left\{(\gamma_{1}f(w_{1}),\gamma_2f(w_2),\ldots,\gamma_{n}f(v_{n})) \mid f(X) \in \mathbb{F}_q[X]_{k}\right\}\text{.}
\end{align}
It is well known  (see \cite{JLLX2010}) that the   ${GRS}_{n,k}(w,\gamma)$ is a linear code with parameters $[n,k,n-k+1]_{q}$ and the Hermitian dual
$({GRS}_{n,k}(w,\gamma))^{\bot_H}$ of ${GRS}_{n,k}(w,\gamma)$ is also a GRS code with  parameters are  $[n,n-k,k+1]_{q}$.
 Moreover, ${GRS}_{n,k}(w,\gamma) \subsetneq {GRS}_{n,k+1}(w,\gamma)$.  By letting  $C_1={GRS}_{n,k}(w,\gamma)$ and $C_2= ({GRS}_{n,k+i}(w,\gamma))^{\perp_H}$  (with  $0\leq i\leq n-k$) in Corollary \ref{lasther}, we have the following  corollary.
 
\begin{corollary}\label{cor5.2}  Let $r$ be a prime power and let $q=r^2$. Let $0\leq k\leq n \leq q$ be integers. Then there exists a Hermitian self-orthogonal matrix-product code with parameters $[2n, n-i, d]_q$ for all $0\leq i \leq n-k$, where $d\geq \min\{2(n-k+1), k+i+1\}$.
    \end{corollary}

Some examples of good Hermitian self-orthogonal codes over small finite fields derived from Corollary \ref{cor5.2}  of  length $2q$ are given in Table \ref{T2}.

\renewcommand{\multirowsetup}{\centering} 
\setlength{\tabcolsep}{10pt}
\begin{table}[!hbt]
    \caption{Hermitian Self-Orthogonal Matrix-Product Codes over $\mathbb{F}_q$}
    \label{T2}
    \begin{center}
        \begin{tabular}{|c|c|c|c|l|}
            \hline 
            $q$&$n$&$k$&$i$& Parameters  \\
            \hline
           $ 9 $&$ 9 $&$ 5 $&$ 4 $&$[ 18 , 5 ,d]_{ 9 } \text{ with }d\geq  10 $\\
          & &$ 6 $&$ 1 $& $[ 18 , 8 ,d]_{ 9 } \text{ with }d\geq  8 $\\ \hline
           $ 16 $&$ 16 $&$ 9 $&$ 6 $& $[ 32 , 10 ,d]_{ 16 } \text{ with }d\geq  16 $\\
            & &$ 10 $&$ 3 $& $[ 32 , 13 ,d]_{ 16 } \text{ with }d\geq  14 $\\
            & &$ 11 $&$ 0 $& $[ 32 , 16 ,d]_{ 16 } \text{ with }d\geq  12 $\\
           \hline
           $ 25 $&$ 25 $&$ 13 $&$ 12 $& $[ 50 , 13 ,d]_{ 25 } \text{ with }d\geq  26 $\\
           & &$ 14 $&$ 9 $& $[ 50 , 16 ,d]_{ 25 } \text{ with }d\geq  24 $\\
           &&$ 15 $&$ 6 $& $[ 50 , 19 ,d]_{ 25 } \text{ with }d\geq  22 $\\
           &&$ 16 $&$ 3 $& $[ 50 , 22 ,d]_{ 25 } \text{ with }d\geq  20 $\\
           &&$ 17 $&$ 0 $&$[ 50 , 25 ,d]_{ 25 } \text{ with }d\geq  18 $\\
           \hline
           $ 49 $&$ 49 $&$ 25 $&$ 24 $& $[ 98 , 25 ,d]_{ 49 } \text{ with }d\geq  50 $\\
           &&$ 26 $&$ 21 $& $[ 98 , 28 ,d]_{ 49 } \text{ with }d\geq  48 $\\
         &&$ 27 $&$ 18 $& $[ 98 , 31 ,d]_{ 49 } \text{ with }d\geq  46 $\\
           &&$ 28 $&$ 15 $& $[ 98 , 34 ,d]_{ 49 } \text{ with }d\geq  44 $\\
         &&$ 29 $&$ 12 $& $[ 98 , 37 ,d]_{ 49 } \text{ with }d\geq  42 $\\
          &&$ 30 $&$ 9 $& $[ 98 , 40 ,d]_{ 49 } \text{ with }d\geq  40 $\\
          &&$ 31 $&$ 6 $& $[ 98 , 43 ,d]_{ 49 } \text{ with }d\geq  38 $\\
           &&$ 32 $&$ 3 $&$[ 98 , 46 ,d]_{ 49 } \text{ with }d\geq  36 $\\
           &&$ 33 $&$ 0 $& $[ 98 , 49 ,d]_{ 49 } \text{ with }d\geq  34 $\\
           \hline
        \end{tabular}
    \end{center}
\end{table}


For the special case where  $i=0$, or equivalently, $C_1={GRS}_{n,k}(w,\gamma)$ and $C_2=C_1^{\perp_H}= ({GRS}_{n,k}(w,\gamma))^{\perp_H}$,  a Hermitian self-dual matrix-product code  can be constructed via  Corollaries  \ref{cor3.3} and \ref{cor5.2}.

\begin{corollary} Let $r$ be a prime power and let $q=r^2$. Let $0\leq l\leq n \leq q$ be integers. Then there exists a Hermitian self-orthogonal matrix-product code with parameters $[2n, n, d]_q$ with $d\geq \min\{2(n-k+1),k+1\}$.
\end{corollary}


\section{Conclusion and Remarks}
The well-known  matrix-product construction for linear codes has been applied to construct Hermitian self-orthogonal codes.   Criterion for the  underlying matrices and the  input codes required in the constructions have been determined.  In many cases,    the Hermitian self-orthogonality of the input codes and the  assumption that the underlying matrix is unitary  can be  relaxed.  Illustrative  examples of good Hermitian self-orthogonal codes have been given as well.

Some special matrices used in the constructions such as weakly quasi-unitary  and weakly anti-quasi-unitary matrices have been given in some cases. In general, the study of a matrix $A \in M_{s,l}(\mathbb{F}_q)$  such that $AA^{\dag}$ is diagonal or  anti-diagonal is also an interesting problem. 

 For applications, it is well know that Hermitian self-orthogonal codes can be applied   in constructing symmetric quantum codes (see, for example,  \cite{JLLX2010}, \cite{JX2012}, \cite{ZG2015},  and \cite{LLY2016}).  Hence, the   codes obtained  in this paper can  be  applied in the such constructions as well.  
 
\section*{Acknowledgment}
This research was supported by the Thailand Research Fund under Research
Grant MRG6080012.

%
%
%
%
%
%

\end{document}